\documentclass[journal]{IEEEtran}

\usepackage{amssymb}
\usepackage{caption2}
\usepackage{amsmath}
\usepackage{multirow}
\usepackage{amsthm}

\newtheorem{theorem}{Theorem}[section]

\newtheorem{lemma}[theorem]{Lemma}

\newtheorem{corollary}[theorem]{Corollary}

\newcommand{\ml}{\mathcal}

\newcommand{\s}{\subseteq}

\newcommand{\bi}{\binom}
\newcommand{\fr}{\frac}
\newcommand{\lc}{\lceil}
\newcommand{\rc}{\rceil}
\begin{document}

\title{New bounds on the number of tests for disjunct matrices}
\author{Chong Shangguan and Gennian Ge
\thanks{The research of G. Ge was supported by the National Natural Science Foundation of China under Grant Nos. 11431003 and 61571310.
}
\thanks{C. Shangguan is with the School of Mathematical Sciences, Zhejiang University,
Hangzhou 310027,  China (e-mail: theoreming@163.com).}
\thanks{G. Ge is with the School of Mathematical Sciences, Capital Normal University,
Beijing 100048, China (e-mail: gnge@zju.edu.cn). He is also with Beijing Center for Mathematics and Information Interdisciplinary Sciences, Beijing, 100048, China.}
\thanks{Copyright (c) 2014 IEEE. Personal use of this material is permitted. However, permission to use this material for any other purposes must be obtained from the IEEE by sending a request to pubs-permissions@ieee.org.}
}
\maketitle

\begin{abstract}Given $n$ items with at most $d$ of which being positive, instead of testing these items individually, the theory of combinatorial group testing aims to identify all positive items using as few tests as possible. This paper is devoted to a fundamental and thirty-year-old problem in the nonadaptive group testing theory. A binary matrix is called $d$-disjunct if the boolean sum of arbitrary $d$ columns does not contain another column not in this collection. Let $T(d)$ denote the minimal $t$ such that there exists a $t\times n$ $d$-disjunct matrix with $n>t$. $T(d)$ can also be viewed as the minimal $t$ such that there exists a nonadaptive group testing scheme which is better than the trivial one that tests each item individually. It was known that $T(d)\ge\bi{d+2}{2}$ and was conjectured that $T(d)\ge(d+1)^2$. In this paper we narrow the gap by proving $T(d)/d^2\ge(15+\sqrt{33})/24$, a quantity in [6/7,7/8].

\end{abstract}


\section{Introduction}

Given $n$ items with at most $d$ of which being positive, instead of testing these items individually, the theory of combinatorial group testing aims to identify all positive items using as few tests as possible. Its history can date back to World War \uppercase\expandafter{\romannumeral2} when the biologists needed to identify people with syphilitic antigen from a large population. The idea was first introduced by Dorfman for testing blood samples \cite{dorfman1943detection}. Since then, the theory of group testing has been extensively studied due to its many applications in a variety of fields, such as chemical leak testing \cite{chemicalleaking}, electric shorting detection \cite{chen1989detecting}, multi-access channel communication \cite{multiaccess1}, \cite{multiaccess2}, DNA screening \cite{du2006pooling}, pattern finding \cite{patternfinding} and recently, network security \cite{networksecurity}.

Assume each of the $n$ items is associated with an undetermined binary status, positive (used to be called defective) or negative (used to be called pure). A test can be viewed as a subset of the items. The outcome of a test is $1$ (positive) whenever it contains a positive item and $0$ (negative) otherwise. The problem is to identify all positive items. Our strategy is that we group the items into several tests. In each test, a positive outcome indicates that at least one of the items included in this test is positive and a negative outcome indicates that all items included are negative. Usually the number of positive items is bounded by a positive integer $d$. To find this specified subset, one can trivially test every item individually and in this way $n$ tests are needed, which will be a waste if $d$ is much smaller than $n$. On the other hand the information-theoretic bound suggests that at least $\log\sum_{i=0}^d\bi{n}{i}\approx d\log\fr{n}{d}$ tests are needed. There is a huge gap between $n$ and $d\log\fr{n}{d}$, so we should carefully design our testing algorithms.

In general there are two types of algorithms, namely, adaptive (sequential) or nonadaptive. An adaptive algorithm is designed to several rounds and the later tests are allowed to use the outcomes of all previous ones. Conversely, a nonadaptive algorithm carries out all tests simultaneously and all positive items should be identified in a single round. The adaptive algorithms inherently require fewer tests than the nonadaptive ones since more information can be used. Asymptotically, for a nonadaptive group testing scheme, the known bounds show that at least $\Omega(\fr{d^2}{\log d}\log n)$ tests are needed \cite{DR1}, \cite{furedi}, \cite{lemmaofcff}. But for the adaptive setting, there exist algorithms with as small as $O(d\log n)$ tests, optimal up to a constant factor \cite{twostage}. However, the nonadaptive algorithms have their own advantages. They are time-saving and are encouraged in the applications that time is the most emergent issue, such as DNA screening and network security.

The application of nonadaptive group testing into molecular biology, especially in the design of screening experiment, has been extensively studied during the last twenty years. The readers are referred to the comprehensive book of Du and Hwang \cite{du2006pooling} for more detailed information. Recently, Xuan et. al. \cite{networksecurity} have found that the idea of group testing can be adapted naturally to network security. In the simplest attack scenario there are $n$ clients connecting to $t$ servers and among $n$ clients, there are $d$ attackers. Just like the detecting of positive response in a single test, once an attacker starts attacking a server, the resources of this server will be exhausted dramatically. So it is not hard to identify which server is a victim. In the security setting the nonadaptive group testing attracts more attention since it is very important to detect the defective items as soon as possible before they cause great damage to the whole network.


A nonadaptive group testing scheme can be represented as a $t\times n$ boolean matrix $M$ whose rows are indexed by the tests and whose columns are indexed by the items, in which $M(i,j)=1$ if the $j$-item is contained in the $i$-th test and 0 otherwise. The matrix $M$ is often designed to be a disjunct matrix. The notion of ``disjunctness" was introduced by Kautz and Singleton \cite{KS} when they were studying some important problems in information retrieval system. Later Erd\H{o}s, Frankl and F{\"u}redi \cite{CFF} introduced a combinatorial object named as ``cover-free family", whose incidence matrix was exactly a disjunct matrix.
We say a matrix is $d$-disjunct if the boolean sum of any $d$ columns does not contain any other column. In other words, a matrix is $d$-disjunct if for any $d+1$ columns indexed by $c_1,\ldots,c_{d+1}$ and for every $j=1,\ldots,d+1$, there must exist a row that has exactly one 1 in the $c_j$-th position and zeros in the other positions. It is not very difficult to see that a $d$-disjunct matrix gives a nonadaptive group testing scheme that identifies any positive set up to size $d$. On the other hand any nonadaptive group testing scheme that identifies any positive set up to size $d$ must also be a $(d-1)$-disjunct matrix \cite{duhwang2000}. Denote $t(d,n)$ the minimal $t$ such that a $t\times n$ $d$-disjunct matrix exists. The recent results of D'yachkov et. al. \cite{newbound3}, \cite{newbound2} showed that the following bounds hold asymptotically

\begin{equation*}
\fr{d^2\log n}{2\log d}(1+o(1))\le t(d,n)\le \fr{e^2d^2\log n}{4\log d}(1+o(1)).
\end{equation*}

For general $n$, it holds that $t(d,n)\ge\min\{\bi{d+2}{2},{n}\}$, which was attributed to Bassalygo by D'yachkov and Rykov \cite{DR2}. This bound implies that if $n\le\bi{d+2}{2}$ then no $d$-disjunct algorithm is superior to the trivial algorithm that tests every item individually. An interesting problem is suggested by the above result: given $d$, when does there exist a $d$-disjunct algorithm better than the trivial one? It is equivalent to ask: given $d$, what is the minimal $t$ such that there exists a $t\times n$ $d$-disjunct matrix with $n\ge t+1$. Denote this minimal $t$ by $T(d)$. Obviously, we have $T(d)\ge\bi{d+2}{2}$ and $t(d,n)\ge\min\{T(d),{n}\}$. In 1985, Erd{\"{o}}s, Frankl and F{\"u}redi conjectured that \cite{CFF}


\begin{eqnarray*}
\lim_{d\rightarrow\infty}T(d)/d^2=1\qquad(weaker~version),\\
T(d)\ge(d+1)^2\qquad(stronger~version),
\end{eqnarray*}

\noindent and they stated without proof that the stronger version holds for $d\le3$ and $\lim_{d\rightarrow\infty}T(d)/d^2\ge 5/6$. Note that the incidence matrix of an affine plane of order $d+1$ is a $(d+1)^2\times ((d+1)^2+(d+1))$ $d$-disjunct matrix with constant column weight $d+1$. And an affine plane of order $d+1$ exists if $d+1$ is a prime power \cite{affineplane}. This implies $\lim_{d\rightarrow\infty}T(d)/d^2\le1$ and $T(d)\le(d+1)^2$ when $d+1$ is a prime power. Later, Huang and Hwang \cite{r<=4} proved the stronger version for $d=4$, while Chen and Hwang \cite{r=5} for $d=5$. In this paper, based on a graph matching theorem of Erd\H{o}s and Gallai \cite{maintool}, we show that $T(d)\geq\frac{15+\sqrt{33}}{24}d^2$ by counting the number of specified substructures contained in the columns of the matrix. Our result significantly improves the previous ones. It is also worth mentioning that disjunct matrices with constant column weight are of particular interest in the framework of DNA screening \cite{du2006pooling}. In the thesis of Chee \cite{chee1996turan}, the author considered the above conjecture for $d$-disjunct matrix with constant column weight $d+1$ and the problem was not completely settled (see, Theorem 5.3.1 of \cite{chee1996turan}). By an easy counting argument we will verify the conjecture under this constant weight constraint.

Our main results are presented as follows.

\begin{theorem}\label{specialcase} Suppose $M$ is a $t\times n$ $d$-disjunct matrix with constant column weight $d+1$. If $n>t$, then $t\ge(d+1)^2$.
\end{theorem}

\begin{theorem} \label{main} Suppose $M$ is a $t\times n$ $d$-disjunct matrix. If $n>t$, then it holds that $t\ge\frac{15+\sqrt{33}}{24}d^2$.
\end{theorem}

The rest of this paper is organised as follows. In Section 2 we will prove Theorem \ref{specialcase} and in Section 3 we will prove Theorem \ref{main}. We conclude this paper in Section 4.

\section{A simple bound for constant weight matrix}

For a $t\times n$ binary matrix $M$, 1s and 0s can represent the incidence structure of a corresponding set system. Let $T=\{1,\ldots,t\}$ be a set of $t$ elements and let $\mathcal{F}=\{F_1,\ldots,F_n\}\s 2^T$ be a collection of subsets of $T$. Then $M$ can be viewed as the incidence matrix of $(T,\mathcal{F})$ such that for all $1\le i\le t$, $1\le j\le n$, $i\in F_j$ if and only if $M(i,j)=1$. We can simply replace the set $F_j$ by the column $c_j$, then
we just write $i\in c_j$ if $M(i,j)=1$, which also indicates that row $i$ is contained in column $c_j$.
A column of $M$ is called isolated if there exists a row incident to it but not to any other column. If $M$ is $d$-disjunct and has an isolated column $c$, then by deleting $c$ and the isolated row contained in it we get a $(t-1)\times (n-1)$ matrix $M'$ which maintains the $d$-disjunctness.
Then $n>t$ holds for the original matrix $M$ is equivalent to $n-1>t-1$ holds for the new matrix $M'$. By the definition of $T(d)$, the minimal $t$ satisfying $(n-1)>(t-1)$ is at least $T(d)+1$. We can summarise this observation as the following lemma.

\begin{lemma}\label{easy}
Suppose $M$ is a $t\times n$ $d$-disjunct matrix with an isolated column $c$. If $n>t$, then $t>T(d)$.
\end{lemma}

\begin{proof}
Deleting column $c$ and the corresponding isolated rows yields a $(t-r_c)\times (n-1)$ $d$-disjunct matrix, where $r_c\ge1$ is the number of isolated rows contained in $c$. Then by the definition of $T(d)$ we have $t-r_c\ge T(d)$ and hence $t\ge T(d)+r_c$ since $n-1>t-r_c$.
\end{proof}

Therefore, to determine $T(d)$ we only need to consider the matrices with no isolated columns. The weight of a column $c$, denoted as $|c|$, is defined to be the number of 1s contained in it. One can see that a non-isolated column in a $d$-disjunct matrix has weight at least $d+1$, since any 1 in this column is contained in some other columns. So for a matrix with no isolated columns, the minimal weight of the columns is at least $d+1$. Theorem \ref{specialcase} establishes the validity for the stronger version of the conjecture in the simplest case, i.e., the $d$-disjunct matrix being considered is of constant column weight $d+1$.

We present the proof of Theorem \ref{specialcase} as follows.


\begin{proof}[\textbf{Proof of Theorem \ref{specialcase}}] We can always assume that $M$ has no isolated columns by Lemma \ref{easy}. Then for arbitrary two distinct columns $c,c'$ it is easy to see $|c\cap c'|\le 1$. Denote $C(i)$ as the collection of columns that has a 1 in the $i$-th row.
By counting the number of 1s in the whole matrix we get $\sum_{i=1}^t|C(i)|=n(d+1)\ge(t+1)(d+1)$. Therefore, there exists some $1\le i_0\le t$ such that $|C(i_0)|\ge\lc\fr{(d+1)(t+1)}{t}\rc\ge d+2$. Note that $c\cap c'=\{i_0\}$ holds for all distinct $c,c'\in C(i_0)$, then the theorem follows from $t\ge |\vee_{c\in C(i_0)} c|=1+(d+2)d=(d+1)^2$, where $\vee$ denotes the boolean sum (the union) of the columns.
\end{proof}


\section{A general bound for $T(d)$}

Suppose $K$ is a $k$-element set, we use $\binom{K}{\lambda}$ to denote the collection of all $\lambda$-element subsets of $K$, where $1\leq \lambda\leq k$ is a positive integer. Let $\ml{G}\s\bi{K}{\lambda}$ be a family of $\lambda$-element subsets of $K$. The matching number $v(\ml{G})$ is defined to be the maximum number of pairwise disjoint members of $\ml{G}$. One of the classical problems of extremal set theory is to determine $\max|\ml{G}|$ for fixed $v(K)$.
Define $m(k,\lambda,\mu)=\max\{|\mathcal{G}|:\ml{G}\s \binom{K}{\lambda},~|K|=k,~v(\ml{G})\le\mu\}$. In 1959, Erd{\H{o}}s and Gallai \cite{maintool} determined $m(k,\lambda,\mu)$ for $\lambda=2$ (see, Theorem 4.1 of \cite{maintool})

\begin{lemma}{\rm(\cite{maintool})} \label{maintoo}  $m(k,2,\mu)\leq\max\{\binom{2\mu+1}{2},\bi{k}{2}-\binom{k-\mu}{2}\}$ for $k\ge 2\mu+1$.
\end{lemma}


A very important notion in studying disjunct matrix is ``privateness", which was introduced as ``own part" in \cite{CFF}. For a
given matrix $M$, a subset of $\{1,\ldots,t\}$ is private if it belongs to a unique column. On the contrary, a subset of $\{1,\ldots,t\}$ is called non-private if it belongs to at least two columns. When proving Theorem \ref{specialcase}, we actually consider the private 1-subsets since a column is isolated if and only if it contains a private 1-subset. In order to establish our general bound, we investigate the properties of private 2-subsets. More precisely, a lower bound for the number of private 2-subsets that a column must contain is obtained. For a column $c$, denote $P(c)=\{T\s\{1,\ldots,t\}:|T|=2,~T\s c~and ~T~is~private\}$ as the collection of private 2-subsets contained in $c$ and denote $N(c)$ as the collection of non-private 2-subsets contained in $c$. If column $c$ has weight $k$, then $\bi{k}{2}=|P(c)|+|N(c)|$ since $P(c)$ and $N(c)$ partition all 2-subsets of $c$. The lemma below presents an upper bound for the size of $N(c)$.

\begin{lemma}\label{claim} Suppose $M$ is a $t\times n$ $d$-disjunct matrix with no isolated columns. Then for any arbitrary column $c$ satisfying $|c|=d+s$, where $1\leq s\leq d-1$, it holds that $|N(c)|\leq m(d+s,2,s-1)\leq\max\{\bi{2s-1}{2},\bi{d+s}{2}-\bi{d+1}{2}\}$.
\end{lemma}
\begin{proof} It suffices to show $N(c)$ does not contain $s$ pairwise disjoint members. If otherwise, the left $(d+s)-2s=d-s$ 1s of $c$ is contained in the union of some $d-s$ columns of $M$ since $c$ has no private 1-subsets. Then $c$ is contained in the union of some $s+(d-s)=d$ columns, which violates the $d$-disjunct property.
\end{proof}

For $s\geq 1$, by direct computation one can verify the following formula holds
\begin{equation}
       \begin{aligned}&\max\{\bi{2s-1}{2},\bi{d+s}{2}-\bi{d+1}{2}\}\\&=
       \begin{cases} \bi{d+s}{2}-\bi{d+1}{2}, & s\leq\fr{2}{3}d+\frac{2}{3},\\\bi{2s-1}{2}, & s\geq\fr{2}{3}d+\frac{2}{3}.
       \end{cases}
       \end{aligned}
\end{equation}
One more lemma is needed to prove Theorem \ref{main}.


\begin{lemma}{\label{disjunct}}
Suppose $M$ is a $t\times n$ $d$-disjunct matrix. Assume $c$ is an arbitrary column of $M$ with weight $w_c$, then deleting $c$ and all rows intersecting it yields a $(t-w_c)\times(n-1)$ $(d-1)$-disjunct matrix.
\end{lemma}

\begin{proof}See Lemma 2.2.2 of \cite{du2006pooling}.
\end{proof}

\begin{proof}[\textbf{Proof of Theorem \ref{main}}] Again, we can assume that $M$ has no isolated columns by Lemma \ref{easy}. Then the minimal column weight of $M$ is at least $d+1$. We will apply induction on $d$ to prove the theorem. Our statement is true for $1\le d\le 5$ by previous results.
Assume the statement is true for $d-1$. Let $c$ be the column with the largest column weight and for the sake of simplicity, denote $\kappa=(15+\sqrt{33})/24$. Then our goal is to prove $t\geq \kappa d^2$.
The proof can be divided into two cases:

  \textbf{Case 1.} $|c|\geq\lc 2\kappa d\rc$. By Lemma \ref{disjunct}, deleting $c$ and all rows intersecting it we get a $(t-|c|)\times (n-1)$ $(d-1)$-disjunct matrix. Obviously, $n-1>t-|c|$ since $n>t$. By the induction hypothesis we can deduce that $t\ge |c|+\kappa(d-1)^2\ge 2\kappa d+\kappa(d-1)^2\ge \kappa d^2.$

  \textbf{Case 2.} $|c|\le \lfloor 2\kappa d\rfloor$. Then every column of $M$ has weight at most $\lfloor2\kappa d\rfloor$. Fix a column $u$ with $|u|=d+s$, where $1\leq s\leq(2\kappa-1)d$. Let us estimate the number of private 2-subsets contained in $u$. On one hand, if $|u|\leq\frac{5d}{3}+\frac{2}{3}$, then by the first formula of (1) we have $|P(c)|=\bi{d+s}{2}-|N(c)|\geq\bi{d+1}{2}$. On the other hand, if $|u|>\frac{5d}{3}+\frac{2}{3}$, then $2d/3\leq s\leq(2\kappa-1)d$, by the second formula of (1) we have $|P(c)|\geq\bi{d+s}{2}-\binom{2s-1}{2}\geq(d^2+2ds-3s^2)/2\geq(3\kappa-1)(2-2\kappa)d^2=\kappa d^2/2$. Note that $|P(c)|\ge \kappa d^2/2$ holds in both cases since $\kappa<1$. Then the statement follows from the fact that the number of private 2-subsets in $\{1,\ldots,t\}$ can not exceed $\bi{t}{2}$, i.e., $\bi{t}{2}\ge\sum_{c}|P(c)|\geq n\times \kappa d^2/2\geq(t+1)\kappa d^2/2$.
\end{proof}

The following result is straightforward.

\begin{corollary}
Denote $t(d,n)$ the minimal $t$ such that there exists a $t\times n$ $d$-disjunct matrix. Then it holds that $t(d,n)\ge\min\{\frac{15+\sqrt{33}}{24}d^2,n\}$.
\end{corollary}

\begin{proof}
The corollary holds since $t(d,n)\ge\min\{T(d),n\}$.
\end{proof}

Through a similar argument to that of Theorem \ref{main}, one can prove the following corollary.
\begin{corollary} Suppose $M$ is a $t\times n$ $d$-disjunct matrix. If $n>t$ and for every column $c$ of $M$, there is $|c|\leq\lfloor\fr{5d}{3}\rfloor$. Then it holds that $t>d^2+d+1$.
\end{corollary}
\begin{proof} By (1) we have $|P(c)|\geq \bi{d+1}{2}$ for every non-isolated column $c$. Then the conclusion follows from
$\bi{t}{2}\ge\sum_{c}|P(c)|\geq n\bi{d+1}{2}\geq(t+1)\bi{d+1}{2}$.
\end{proof}

\section{Concluding remarks}

In this paper we consider the lower bound of the minimal $t$ when there exists a $t \times n$ $d$-disjunct matrix with $n > t$,  and our new bound improves the previous results significantly. The novelty of our method is that we consider the properties of private 2-subsets of the given disjunct matrix and apply a graph matching theorem of Erd{\H{o}}s and Gallai \cite{maintool}. A natural idea to generalize our method is to consider larger private subsets and then a hypergraph version of matching theorem will be needed \cite{hypergraphmatching}. It will be interesting if someone can improve our results in this way.

\bibliographystyle{IEEEtranS}
\bibliography{Ref}
\begin{IEEEbiographynophoto}{Chong~Shangguan}
is currently a Ph.D. student at Zhejiang University, Hangzhou,
Zhejiang, P. R. China. His research interests include
extremal combinatorics, coding theory, and
their interactions.
\end{IEEEbiographynophoto}

\begin{IEEEbiographynophoto}{Gennian~Ge}
received the M.S. and Ph.D. degrees in mathematics from Suzhou
University, Suzhou, Jiangsu, P. R. China, in 1993 and 1996,
respectively. After that, he became a member of Suzhou
University. He was a postdoctoral fellow in the Department of
Computer Science at Concordia University, Montreal, QC,
Canada, from September 2001 to August 2002, and a visiting
assistant professor in the Department of Computer Science at
the University of Vermont, Burlington, Vermont, USA, from
September 2002 to February 2004.  He was  a
full professor in the Department of Mathematics at Zhejiang
University, Hangzhou, Zhejiang, P. R. China, from
March 2004 to February 2013.  Currently, he
is a full professor in the School of Mathematical Sciences at  Capital Normal University, Beijing, P. R. China.  His research
interests include the constructions of combinatorial designs
and their applications to codes and crypts.

Dr. Ge is on the Editorial Board of {\em Journal of
Combinatorial Designs, SCIENCE CHINA Mathematics,
Applied Mathematics$-$A Journal of Chinese Universities}. He received the 2006
Hall Medal from the Institute of Combinatorics and its
Applications.
\end{IEEEbiographynophoto}

\end{document}